\documentclass[letterpaper, 10 pt, conference]{ieeeconf}  % Comment this line out if you need a4paper

\IEEEoverridecommandlockouts                              % This command is only needed if 
\usepackage{graphicx,epstopdf}							   % Graphical packages
\usepackage{amsmath,mathrsfs,array,amssymb} % Mathematical packages
\usepackage{hyperref}														   % Changes al references to links, also external references
\usepackage{dsfont,bbold}													 % Allows for \mathds{} and \mathbb{} commands	
\usepackage{ifthen,xstring}																 % Allows for \ifthenelse and \IfEqCase statement respectively in \newcommand
\usepackage{xparse}																 % Allows for \DeclareDocumentCommand, which is an upgraded version of \newcommand. This allows for multiple optional arguments. 
\usepackage{flushend}															 % Used for balancing the columns on the last page

\usepackage{amsthm}
  {\par\medskip\noindent\minipage{\linewidth}}
  {\endminipage\par\medskip}
			
\newcommand{\degr}{\ensuremath{^\circ}C}				   					% Defines the degrees symbol

\newcommand{\pmat}[1]{                                   % pmatrix environment
\begin{pmatrix}
#1
\end{pmatrix}
}	
\newcommand{\abs}[1]{\left\lvert #1 \right\rvert}						% Absolute value
\newcommand{\norm}[1]{\left\lVert #1 \right\rVert}					% Norm value
\newcommand{\diag}[1]{\text{diag}\{\begin{matrix}#1_1,& #1_2,&\cdots, &#1_n\end{matrix}\}} % Diagonal matrix notation from vector. Input is the vector from which the matrix is constructed.
\newcommand{\R}{\mathbb{R}}												 					% Real numbers
\newcommand{\one}{\mathds{1}}																% 1-vector
\newcommand{\zero}{\mathbf{0}}															% 0-vector
\newcommand{\COP}{\text{COP}}											 					% Coefficient of Performance

\newcommand{\ttsup}{\text{sup}}											 					% supply temperature script
\newcommand{\ttscal}{\text{scalar}}											 		% scalar temperature script
\newcommand{\ttout}{\text{out}}											 					% Output temperature script
\newcommand{\ttin}{\text{in}}											 					% input temperature script
\newcommand{\ttret}{\text{ret}}											 				% return script
\newcommand{\ttrem}{\text{rem}}											 				% Removed script
\newcommand{\Tsafe}{T_\text{safe}}											 		% Safe temperature
\DeclareDocumentCommand{\Tout}{O{PlainLetter} O{NoIndex}}{% Output temperature
		\def\LetterToPrint{T}%
		\def\SubscriptToPrint{\ttout}%
		\def\SuperscriptToPrint{test}%
		\ifthenelse{\equal{#2}{NoIndex}}%
				{\def\myCheckSuperscriptToPrint{0}}%
				{\ifthenelse{\equal{#2}{scal}}%
						{\def\SubscriptToPrint{\ttout,\ttscal}}%
						{\def\myCheckSuperscriptToPrint{1}\def\SuperscriptToPrint{#2}}%
				}%		
		\IfEqCase{#1}{%
				{PlainLetter}{}%
				{bar}{\def\LetterToPrint{\bar{T}}}%
				{tilde}{\def\LetterToPrint{\tilde{T}}}%
				{hat}{\def\LetterToPrint{\hat{T}}}%
				{dot}{\def\LetterToPrint{\dot{T}}}%
				{scal}{\def\SubscriptToPrint{\ttout,\ttscal}}%
		}[\def\myCheckSuperscriptToPrint{1}\def\SuperscriptToPrint{#1}]%
		\ifthenelse{\equal{\myCheckSuperscriptToPrint}{1}}%
				{\LetterToPrint_{\SubscriptToPrint}^{\SuperscriptToPrint}}%
				{\LetterToPrint_{\SubscriptToPrint}}%				
}%
\DeclareDocumentCommand{\Tin}{O{PlainLetter} O{NoIndex}}{% Supply temperature
		\def\LetterToPrint{T}%
		\ifthenelse{\equal{#2}{NoIndex}}%
				{\def\SubscriptToPrint{\ttin}}%
				{\ifthenelse{\equal{#2}{scal}}%
									{\def\SubscriptToPrint{\ttin,\ttscal}}%
									{\def\SubscriptToPrint{\ttin,#2}}%
				}%		
		\IfEqCase{#1}{%
				{PlainLetter}{}%
				{bar}{\def\LetterToPrint{\bar{T}}}%
				{tilde}{\def\LetterToPrint{\tilde{T}}}%
				{hat}{\def\LetterToPrint{\hat{T}}}%
				{dot}{\def\LetterToPrint{\dot{T}}}%
				{scal}{\def\SubscriptToPrint{\ttin,\ttscal}}%
		}[\def\SubscriptToPrint{\ttin,#1}]%
		\LetterToPrint_{\SubscriptToPrint}%
}%
\DeclareDocumentCommand{\Tsup}{O{PlainLetter} O{NoIndex}}{% Supply temperature
		\def\LetterToPrint{T}%
		\def\OneToPrint{1}%
		\ifthenelse{\equal{#2}{NoIndex}}%
				{\def\SubscriptToPrint{\ttsup}}%
				{\ifthenelse{\equal{#2}{scal}}%
									{\def\OneToPrint{0}\def\SubscriptToPrint{\ttsup}}%
									{\def\SubscriptToPrint{\ttsup,#2}}%
				}%		
		\IfEqCase{#1}{%
				{PlainLetter}{}%
				{bar}{\def\LetterToPrint{\bar{T}}}%
				{tilde}{\def\LetterToPrint{\tilde{T}}}%
				{hat}{\def\LetterToPrint{\hat{T}}}%
				{dot}{\def\LetterToPrint{\dot{T}}}%
				{scal}{\def\OneToPrint{0}}%
		}[\def\SubscriptToPrint{\ttsup,#1}]%
		\ifthenelse{\equal{\OneToPrint}{1}}%
				{\one\LetterToPrint_{\SubscriptToPrint}}%
				{\LetterToPrint_{\SubscriptToPrint}}%
}%
\newcommand{\Q}[2][0]{																		% Output heat
		\IfEqCase{#2}{%
				{out}{\def\SubscriptToPrint{\ttout}}%
				{in}{\def\SubscriptToPrint{\ttin}}%
				{ret}{\def\SubscriptToPrint{\ttret}}%
				{rem}{\def\SubscriptToPrint{\ttrem}}%
				{sup}{\def\SubscriptToPrint{\ttsup}}%
		}[\def\SubscriptToPrint{#2}]%
		\ifthenelse{\equal{#1}{0}}%
				{Q_\SubscriptToPrint}%
				{Q_{\SubscriptToPrint}^#1}%
}%
\DeclareDocumentCommand{\D}{O{PlainLetter} O{NoIndex}}{% Workload
		\def\SubscriptToPrint{0}
		\def\maxvar{0}
		\def\LetterToPrint{D}%	
		\IfEqCase{#1}{%
				{PlainLetter}{}%
				{bar}{\def\LetterToPrint{\bar{D}}}%
				{tilde}{\def\LetterToPrint{\tilde{D}}}%
				{hat}{\def\LetterToPrint{\hat{D}}}%
				{dot}{\def\LetterToPrint{\dot{D}}}%
				{scal}{}%
				{max}{\def\maxvar{1}}%
		}[\def\SubscriptToPrint{1}]%
		\ifthenelse{\equal{\SubscriptToPrint}{1}}%
				{\LetterToPrint_{#1}}%
				{\ifthenelse{\equal{#2}{NoIndex}}%
						{\ifthenelse{\equal{\maxvar}{1}}%
								{\LetterToPrint_\text{max}}%
								{\LetterToPrint}%
						}%
						{\ifthenelse{\equal{\maxvar}{1}}%
								{\LetterToPrint_{\text{max}}^#2}%
								{\LetterToPrint_{#2}}%
						}%
				}%
}%

\newcommand{\Dmax}[1][NoIndex]{\D[max][#1]}
\newcommand{\Tsupscal}{\Tsup[scal]}	 								% Supply temperature scalar
\def\be{\begin{equation}}
\def\ee{\end{equation}}
\def\ba{\begin{array}}
\def\ea{\end{array}}

\theoremstyle{plain}
\newtheorem{thm}{Theorem}
\newtheorem{lem}{Lemma}

\theoremstyle{definition}
\newtheorem{example}{Example}

\newtheorem{assum}{Assumption}
\newtheorem{property}{Property}
\newtheorem{problem}{Problem}

\theoremstyle{remark}
\newtheorem{rem}{Remark}

\title{Optimized Thermal-Aware Job Scheduling and Control of Data Centers}
\author{Tobias Van Damme, Claudio De Persis and Pietro Tesi%
\thanks{All authors are with the Department of ENTEG, Faculty of Mathematics and Natural Sciences,
				University of Groningen, 9747 AG Groningen, The Netherlands.
				\texttt{ \{t.van.damme, c.de.persis, p.tesi\}@rug.nl}}%
\thanks{The authors declare no competing financial interest}%
}

\begin{document}

\maketitle
\thispagestyle{empty}
\pagestyle{empty}

%%%%%%%%%%%%%%%%%%%%%%%%%%%%%%%%%%%%%%%%%%%%%%%%%%%%%%%%%%%%%%%%%%%%%%%%%%%%%%%%
\begin{abstract}
Analyzing data centers with thermal-aware optimization techniques is a viable approach to reduce energy consumption of data centers. By taking into account thermal consequences of job placements among the servers of a data center, it is possible to reduce the amount of cooling necessary to keep the servers below a given safe temperature threshold. We set up an optimization problem to analyze and characterize the optimal setpoints for the workload distribution and the supply temperature of the cooling equipment. Furthermore under mild assumptions we design and analyze controllers that drive the data center to the optimal state without knowledge of the current total workload to be handled by the data center. The response of our controller is validated by simulations and convergence to the optimal setpoints is achieved under varying workload conditions. 
\end{abstract}

%%%%%%%%%%%%%%%%%%%%%%%%%%%%%%%%%%%%%%%%%%%%%%%%%%%%%%%%%%%%%%%%%%%%%%
% Introduction
%%%%%%%%%%%%%%%%%%%%%%%%%%%%%%%%%%%%%%%%%%%%%%%%%%%%%%%%%%%%%%%%%%%%%%

\section{Introduction}
Data centers are big energy consumers, in 2013 data centers consumed 350 billion kWh of energy, 1.73\% of the global electricity consumption \cite{DCD14JournalJanuaryFebruaryPage1617,ElectricityConsumptionWorld}. With the world being digitized more and more each year, this number is likely to increase as well. Therefore in the last decade computer scientists and control engineers have made efforts to reduce the energy consumption of data centers by devising methods to increase the operational efficiency of these computer halls \cite{Hameed14Survey}. 

Although much progress has been made, there are still several challenges ensuring efficient operation of the cooling equipment. Due to bad design or unawareness for the thermal properties of the data center, local thermal hotspots can arise. This causes the cooling equipment to overreact to ensure that the temperature of the equipment stays below the safe thermal threshold. These peaks cause the cooling equipment to consume more energy then would be necessary if these hotspots were avoided. Therefore having a good understanding of the thermodynamics involved is vital to increasing the cooling efficiency of the data center.  

To tackle these challenges researchers have studied strategies which uses the knowledge of the thermal properties of the data center to make more intelligent choices how to schedule incoming jobs \cite{Moore05Making, Tang08Energy}. With heuristic methods they showed improvements of up to 30\% less energy consumption with respect to non thermal-aware job schedulers. Other approaches include considering a heterogeneous data center \cite{Sun14Energy} and using these asymmetric properties to analyze trade-offs between performance- and energy-aware algorithms, or distinguishing between different type of jobs when scheduling the load \cite{Jiang14Thermal}. Server consolidation is a natural extension where on top of thermal scheduling, racks are switched on and off to save power. These algorithms usually contain two steps, first to calculate the necessary number of racks and secondly the correct workload scheduling \cite{Moore05Making,Pahlavan14Power,Pakbaznia10Temperature,Li12Joint,Tang06Thermal}.

On the other hand, studies have also been done in a more theoretical direction. Cast as a control problem \cite{Vasic10Thermal} has proposed a control algorithm that tries to maintain the temperature of the equipment around a target value. In \cite{Yin14Adaptive} a two-step algorithm is proposed that first minimizes the energy consumption by estimating the required amount of servers to handle the expected workload. In the second step the algorithm maximizes the response time given a number of servers at its disposal. In an attempt to address scalability a distributed approach has been studied in \cite{Doyle13Distributed}. In this work units, which range from servers to complete data centers, communicate directly and try to achieve a uniform temperature profile. Another distributed control approach in a hybrid systems setting is proposed in \cite{Albea14Hybrid}. The hybrid controller tries to evenly divide the total load among the agents in the network in a distributed fashion.

While all this work has strong points on its own, to the authors best knowledge a thorough analysis and characterization of an energy minimal solution combined with a straightforward control strategy which handles both cooling and job scheduling simultaneously has not been done before. The objective of this work is to supply an easily extendable framework that allows for a characterization of an energy-minimal operating point and then supply straightforward methods for operating the data center such that this operating point is achieved for all load conditions. In addition it should be extendable to include more complex concepts, like switching on and off servers or including quality-of-service constraints. 

The contribution of this work is two-fold. First from existing thermodynamical principles we set up a thermodynamical model from which we derive an optimization problem that combines energy minimization with the thermodynamics. In addition to only including temperature constraints \cite{Li12Joint} we extend the model to also incorporate workload constraints, which allows us to better characterize energy minimal solutions. This design allows for natural extendability to more complicated scheduling policies like switching servers on and off.

Second we develop a novel control strategy for handling the control of the cooling equipment and the workload scheduling simultaneously. Both these control goals have been studied before \cite{Vasic10Thermal,Parolini12Cyber}. However in \cite{Vasic10Thermal} the two control goals were handled separately; In \cite{Parolini12Cyber} a combined algorithm was suggested but due to complexity could lead to non-optimal solution. In contrast our model shows an easy method for handling coordinated cooling and job scheduling control which is guaranteed to converge to the energy minimal solution. Our method is inspired by results from \cite{Burger15Dynamic} where regulation to optimal steady solutions in the presence of disturbances was considered. Therefore our strategy also allows for varying and unknown workload changes while guaranteeing convergence to the energy-minimal operating point.

The remainder of this paper is organized as follows. In \autoref{sec:SystMod} the basic thermodynamics are formulated. Then an optimization problem is formulated in \autoref{sec:ProbForm} and the equivalence to a reduced form is proven. Following up the optimal solution is analytically analyzed and characterized for different load conditions in \autoref{sec:CharacOptSol}. Using this analytical solution a controller is proposed in \autoref{sec:Controller} that can handle unknown load conditions. Finally in \autoref{sec:CaseStudy} a case study is considered to show the performance of the controllers.

An abbreviated version of this paper is submitted to IFAC 2017 for presentation.

\textbf{Notation:} We denote by $\mathbb{R}$ and $\mathbb{R}_{\geq 0}$ the set of real numbers and non-negative real numbers respectively. Vectors and matrices are denoted by $x \in \R^n$ and $A \in \R^{n\times m}$ respectively, the transpose is denoted by $x^T$ and the inverse of a matrix is denoted by $A^{-1}$. If the entries of $x$ are functions of time then the element-wise time derivative is denoted by $\dot{x}(t) := \frac{d}{dt}x$. By $x_i$ we denote the $i$-th element of $x$ and by $a_{ij}$ we denote the $ij$-th element of $A$. If a variable already has another subscript then we switch to superscripts to denote individual elements, i.e. $\Tout[i]$ and $C_3^{ij}$. We write the diagonal matrix constructed from the elements of vector $x$ as $\diag{x}$. The identity matrix of dimension $n$ is denoted by $I_n$, the vector of all ones by $\one \in \R^n$ and the vector of all zeros by $\zero \in \R^n$. Furthermore the vector comparison $x~\preccurlyeq~y$ is defined as the element-wise comparison $x_i~\leq~y_i$ for all elements in $x$ and $y$. Finally a data center consists of $n$ racks.

%%%%%%%%%%%%%%%%%%%%%%%%%%%%%%%%%%%%%%%%%%%%%%%%%%%%%%%%%%%%%%%%%%%%%%
% Section System model
%%%%%%%%%%%%%%%%%%%%%%%%%%%%%%%%%%%%%%%%%%%%%%%%%%%%%%%%%%%%%%%%%%%%%%

\section{System model}
\label{sec:SystMod}
Real life data centers are organized in aisles with many racks each containing a multitude of servers. The cooling of data centers is usually done by air conditioning, therefore the racks are set up in a hot- and cold-aisle configuration. Cold air supplied by the computer room air conditioning (CRAC) units is blown into the cold aisles. The air goes through the racks where it absorbs the heat produced by the servers. The air exits the servers in the hot aisle and is recirculated back to the CRAC units where it is cooled down to the desired supply temperature. 
A scheduler divides incoming tasks among the racks according to some decision policy. The energy consumption of a rack depends on the amount of tasks it is given. By thermodynamical principles almost all of this energy consumption is dissipated as heat in the rack. Ideally all of the exhaust air of the racks is returned to the CRAC, however due to the complex nature of air flows within the data center some of the hot air is recirculated back into the cold aisles. This causes the temperature of the air at the inlet of the racks to rise, creating inefficiencies in the cooling of the data center.

\subsection{Workload}
Requests arriving at the data center are collected by a scheduler which then decides according to some policy how to divide this work among the available racks. We assume that each job has an accompanying tag which denotes the time and the number of computing units (CPU) it requires for execution. Let $J$ denote the integer number of jobs that the scheduler has to schedule in the data center at time $t$. Then $\mathcal{J}(t) = \{1,\cdots,J\}$ denotes the set of jobs to be scheduled at time $t$. Furthermore let $\lambda_j$ be the number of CPU's that job $j$ requires at time $t$. Then the total number of CPU's, $D^*$, the scheduler has to divide over the racks at time $t$ is given by
\begin{equation}
D^*(t) = \sum_{j=1}^{\mathcal{J}(t)}{\lambda_j}.
\label{eq:totWork}
\end{equation}
We denote by $D_i(t)$ the number of CPU's the schedulers assigns to rack $i$ at time $t$. This variable is collected in the vector
\[
D(t) := \pmat{D_1(t) & D_2(t) & \cdots & D_n(t)}^T.
\]
 
\subsection{Power consumption of racks}
\label{subsec:PowConsRacks}
The most common way to model the power consumption of a single rack is using a linear model \cite{Tang06SensorBased,Heath06Mercury,Ranganathan06Ensemble}. In this way the power consumption, $P_i(t)$, of a rack is modeled to consist of a load-independent part, e.g. the equipment consumes a constant amount of power, and a load-dependent part, e.g. the number of CPU's that are actively processing jobs
\begin{equation}
P_i(t) = v_i + w_iD_i(t),
\label{eq:P_i}
\end{equation}
where $v_i$~[Watts] is the power consumption for the unit being powered on, $w_i$~[Watts~CPU$^{-1}$] is the power consumption per CPU in use. The variables are collected in the vectors
\begin{align*}
P(t) &:= \pmat{P_1(t) & P_2(t) & \cdots & P_n(t)}^T,\\
V &:= \pmat{v_1 & v_2 & \cdots & v_n}^T,
\end{align*}
and
\[
W := \diag{w},
\]
so that 
\begin{equation}
P(t) = V+WD(t).
\label{eq:powerVector}
\end{equation}

\subsection{Thermodynamical model}
\label{subsec:ThermModel}
\begin{figure}
	\centering
		\includegraphics{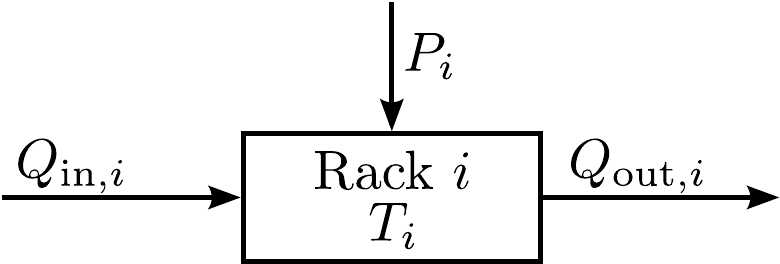}
	\caption{Heat model of an individual rack. $\Q[i]{in}$ is the heat entering the rack, $\Q[i]{out}$ is the heat exiting the rack and $P_i$ is the power consumption of the rack.}
	\label{fig:heatmodel}
\end{figure}

Following similar arguments as in \cite{Vasic10Thermal} and \cite{Tang06SensorBased}, a thermodynamical model for each individual rack is constructed. In \autoref{fig:heatmodel} a graphical representation of the heat flows involved is given. The change of temperature of a rack is given by the difference in heat entering and exiting the rack,	
\begin{equation}
m_ic_p\frac{d}{dt}\Tout[i](t) = \Q[i]{in}(t) - \Q[i]{out}(t) + P_i(t).
\label{eq:heatflow}
\end{equation}
Here $\Tout[i]$~[\degr] is the temperature of the exhaust air at rack $i$, $c_p$~[J~\degr$^{-1}$~kg$^{-1}$] is the specific heat capacity of air, $m_i$~[kg] is the mass of the air inside the rack, $\Q[i]{in}$~[Watts] and $\Q[i]{out}$~[Watts] are the heat entering and exiting the rack respectively. The heat that enters a rack consists of two parts due to the complex air flows in the data center, i.e. the recirculated air originating from the other racks and the cooled air supplied by the CRAC
\begin{equation}
\Q[i]{in}(t) = \sum_{j=1}^n{\gamma_{ji}\Q[i]{out}(t)} + \Q[i]{sup}(t).
\label{eq:heatin}
\end{equation}
Here $\Q[i]{sup}$ [Watts] is the heat supplied by the CRAC to rack $i$, and $\gamma_{ji}$ is the percentage of the flow which recirculates from rack $j$ to rack $i$. 
The relation between heat and temperature is given by
\begin{equation}
Q(t) = \rho c_p fT(t),
\label{eq:heattoT}
\end{equation}
where $\rho$ [kg m$^{-3}$] is the density of the air and $f$ [m$^3$ s$^{-1}$] is the flow rate of the given flow. Combining \eqref{eq:heatin} and \eqref{eq:heattoT} with \eqref{eq:heatflow} yields
\begin{align}
\frac{d}{dt}&\Tout[i](t) = \frac{\rho}{m_i}\left(\sum_{j=1}^n{\gamma_{ji} f_j\Tout[j](t)}  - f_i\Tout[i](t)\right)\nonumber\\
&\quad + \frac{\rho}{m_i}\left(f_i-\sum_{j=1}^n{\gamma_{ji}f_j}\right)\Tsupscal(t) + \frac{1}{m_ic_p}P_i(t),
\end{align}
where $\Tsupscal$ [\degr] is the temperature of the air supplied by the CRAC and $f_i$ is the velocity of the air flow through rack $i$. Rewriting the above relation in matrix form, i.e. combining the temperature changes of all racks in one equation, results in
\begin{equation}
\frac{d}{dt}\Tout(t) = A(\Tout(t)-\Tsup(t)) + M^{-1}P(t).
\label{eq:ThermodynamicalModel}
\end{equation}
Here
\[
\Tout(t) := \pmat{\Tout[1](t) & \Tout[2](t) & \cdots & \Tout[n](t)}^T,
\]
and
\begin{align*} 
A &:= \rho c_pM^{-1}(\Gamma^T - I_n)F, \\
F &:= \diag{f}, \\
M &:= \diag{c_pm},\\
\Gamma &:= [\gamma_{ij}]_{n \times n}.
\end{align*}

\subsection{Power consumption of CRAC}
The power consumption of the CRAC is dependent on the temperature of the air which is returned to CRAC and the supply temperature it has to provide. The air flow which is returned from rack $i$ to the CRAC is given by
\begin{equation}
f_{\ttsup,i}^\ttret = \left(1 - \sum_{j=1}^n{\gamma_{ij}}\right)f_i,
\label{eq:retAirFlow}
\end{equation}
and therefore the heat returned from all the racks to the CRAC is
\begin{equation}
\Q{ret}(t) = \rho c_p \sum_{i=1}^n{\left(1 - \sum_{j=1}^n{\gamma_{ij}}\right)f_i\Tout[i](t)}.
\label{eq:retHeat}
\end{equation}
The heat the CRAC sends back to the data center is given by $\Q{sup}(t) = \rho c_p f_\ttsup \Tsupscal(t)$. With this the heat the CRAC has to remove from the air, $\Q{rem}(t)$, is given by
\begin{align}
\Q{rem}(t) &= \Q{ret}(t) - \Q{sup}(t) \nonumber\\
&=\rho c_p \sum_{i=1}^n{\left[\left(1 - \sum_{j=1}^n{\gamma_{ij}}\right)f_i(\Tout[i](t) - \Tsupscal(t))\right]}\nonumber\\
&= -\one^TMA(\Tout(t)-\Tsup(t)).
\label{eq:heatRem}
\end{align}
To determine the amount of work the CRAC has to do to remove a certain amount of heat, Moore et al. \cite{Moore05Making} introduced the Coefficient of Performance, \COP{}$(\Tsupscal(t))$, to indicate the efficiency of the CRAC as a function of the target supply temperature. They found that CRAC units work more efficiently when the target supply temperature is higher. The \COP{} represents the ratio of heat removed to the amount of work necessary to remove that heat. For a water-chilled CRAC unit in the HP Utility Data Center they found that the \COP{} is a quadratic, increasing function. In a general sense the \COP{} can be any monotonically increasing function. The power consumption of the CRAC units can then be given by

\begin{equation}
P_{AC}(\Tout(t),\Tsupscal(t)) = \frac{\Q{rem}(t)}{\COP(\Tsupscal(t))}.
\label{eq:Pac}
\end{equation}
\begin{assum}
\label{assum:COP}
The \COP{}$(\Tsupscal(t))$, of the CRAC unit considered in this paper, is a monotonically increasing function in the range of operation for $\Tsupscal$.\hfill\QED
\end{assum}
\begin{example}
Let us consider a small example to illustrate the influence of a small difference in supply temperature on the power consumption of the CRAC. Consider the quadratic \COP{}$(\Tsupscal(t))$ found by \cite{Moore05Making}, and two cases where the returned air has to be cooled down by 5 \degr{}, in the first case from 25 \degr{} to 20 \degr{} and in the second case from 30\degr{} to 25\degr{}. Assume that the energy contained in 5 \degr{} temperature difference of air is 100 Watts. In the first case $\COP(20) = 3.19$ and in the second case $\COP(25) = 4.73$. By \eqref{eq:Pac}, the energy consumed by the CRAC to cool down the returned air to the required temperature is
\[
P_{AC,1} = \frac{100}{3.19} = 31.34\text{ W}, \quad P_{AC,2} = \frac{100}{4.73} = 21.14\text{ W}.
\]
Here it seen that if the temperature of the returned air increases by 5 \degr{} the power consumption of the CRAC unit decreases by 30\%. \hfill\QED
\end{example}
Having completed the model finally allows us to formulate the control problem we would like to solve. 

%%%%%%%%%%%%%%%%%%%%%%%%%%%%%%%%%%%%%%%%%%%%%%%%%%%%%%%%%%%%%%%%%%%%%%
% Section Problem formulation
%%%%%%%%%%%%%%%%%%%%%%%%%%%%%%%%%%%%%%%%%%%%%%%%%%%%%%%%%%%%%%%%%%%%%%

\section{Problem formulation}
\label{sec:ProbForm}
The objective of this paper is two-fold, first we want to find optimal setpoints for the temperature distribution, the supply temperature and workload distribution that minimize the power consumption of the data center. Secondly we want to design controllers which ensure convergence of the variables to the obtained setpoints. Hence the control problem is defined as follows:
\begin{problem}
For system \eqref{eq:ThermodynamicalModel} design controllers for the workload distribution $D(t)$ and supply temperature $\Tsupscal(t)$ such that, given an unmeasured total load $D^*(t)$, any solution of the closed-loop system is bounded and satisfies
\begin{align}
\lim_{t\rightarrow\infty}{(\Tout(t) - \Tout[bar])} &= 0,\label{goal:Tout}\\
\lim_{t\rightarrow\infty}{(\Tsupscal(t) - \Tsup[bar][scal])} &= 0,\label{goal:Tsup}\\
\lim_{t\rightarrow\infty}{(D(t) - \bar{D})} &= 0,\label{goal:Work}
\end{align}
where $\Tout[bar]$, $\Tsup[bar][scal]$ and $\bar{D}$ are the optimal setpoint values for the temperature distribution, supply temperature and the power consumption, i.e. workload distribution, respectively, which are defined in \autoref{subsec:OptProb}.\hfill\QED
\end{problem} 
From this point on we will implicitly assume the dependence of the variables on time and only denote it there where confusion might arise otherwise. 

\subsection{Optimization problem}
\label{subsec:OptProb}
We formulate an optimization problem to minimize the power consumption while taking into account the physical constraints of the equipment, i.e the servers only have finite computational capacity and the temperature of the servers cannot exceed a certain threshold. The power consumption of the data center can be written as a combination of 2 parts, the power consumption of the cooling equipment and the power consumption of the racks. Combining \eqref{eq:powerVector} and \eqref{eq:Pac} we can write the total power consumption as
\begin{equation}
\mathcal{C}(\Tout,\Tsupscal,D) = \frac{\Q{rem}}{\COP(\Tsupscal)} + \one^TP(D).
\label{eq:costFunc}
\end{equation}

A reasonable way \cite{Li12Joint,Yin14Adaptive} to formulate the optimization problem is
\begin{subequations}
\label{eq:optimizationDef}
\begin{align}
\min_{\Tout,\Tsupscal,D}\quad& \frac{\Q{rem}}{\COP(\Tsupscal)}+\one^TP(D)\label{eq:costFunctionOpt}\\
s.t. \quad &D^* = \one^TD \label{cstr:totWork}\\
&\zero \preccurlyeq D \preccurlyeq \Dmax \label{cstr:compCap}\\
&\zero = A(\Tout-\Tsup) + M^{-1}P(D)\label{cstr:steadState}\\
& \Tout \preccurlyeq \Tsafe\label{cstr:Tthresh}.
\end{align}
\end{subequations}

Equation \eqref{cstr:totWork} ensures that all the available work is divided among the racks, \eqref{cstr:compCap} encompasses the computational capacity of the rack, i.e. rack $i$ has $\Dmax[i]$ CPU's available at most. The system dynamics should be at steady state once the optimal point has been reached, see \eqref{cstr:steadState}, and finally \eqref{cstr:Tthresh} enforces that the temperature of the racks is below the given safe threshold, $\Tsafe \in \R^n$.

\subsection{Reduced optimization problem}
Due to the non-linear nature of how the \COP{} affects the power consumption it is not trivial to analyze this problem. However under some mild assumptions it is possible to reduce the optimization defined in \eqref{eq:optimizationDef} to a simpler problem. 
\begin{thm}
\label{thm:reducedOpt}
Let the data center consist of homogeneous racks, i.e. $v_i = v$, and $w_i = w$ for all $i$ and assume constraint \eqref{cstr:steadState} is satisfied. Then problem \eqref{eq:optimizationDef} is equivalent to
\begin{subequations}
\label{eq:optimizationReduced}
\begin{align}
\max_{\Tout}\quad& C_1^T\Tout\\
s.t. \quad & \zero\preccurlyeq C_3\Tout + C_4(D^*)\preccurlyeq \Dmax \label{cstr:compCapReduced}\\
& \Tout \preccurlyeq \Tsafe,
\end{align}
\end{subequations}
for suitable $C_1, C_3$ and $C_4$.\hfill\QED
\end{thm}
Before we prove this theorem we need to introduce some notation and extra theory.
\begin{lem}
Equations \eqref{eq:heatRem} and \eqref{cstr:steadState} imply that the following relation holds
\[
\one^TP(D) = -\one^TMA(\Tout-\Tsup)=\Q{rem},
\] 
which reduces the cost function to 
\begin{equation}
\mathcal{C}(\Tout,\Tsupscal,D) = \left(1 + \frac{1}{\COP(\Tsupscal)}\right)\one^TP(D).
\label{eq:costFuncReduced}
\end{equation}
\begin{proof}
By multiplying \eqref{cstr:steadState} by $\one^T M$ and solving for $\one^TP(D)$ we obtain above result.
\end{proof}
\end{lem}
\begin{rem}
\label{remark:totPower}
In many real-life data centers most of the equipment is identical, i.e. the power consumption characteristics of the computational equipment is identical, that is $v_i = v$ and $w_i = w$ for all $i$ in \eqref{eq:P_i}. In this case the data center is said to be composed of homogeneous racks or, more simply, the data center is homogeneous. In case of a homogeneous data center the power consumption is given by $P(D) = v\one + wD$ and the total computational power consumption is given by
\begin{equation}
\one^T P(D) = nv + w\one^TD = nv + wD^*.
\label{eq:PowConsHomogeneous}
\end{equation}

The computational power consumption no longer depends on the way the jobs are distributed but only depends on the total workload. This property simplifies the cost function defined \eqref{eq:costFuncReduced} considerably.\hfill\QED
\end{rem} 
\begin{lem}
\label{lemma:Tsup}
If \eqref{cstr:totWork} and \eqref{cstr:steadState} are satisfied, then there is a unique supply temperature which follows from the desired, chosen temperature distribution, namely
\begin{align}
\Tsupscal &= C_1^T\Tout + C_2(D^*),\label{eq:optimalTsup}\\[0.5ex]
C_1^T &\overset{\Delta}{=}: \frac{\one^TW^{-1}MA}{\one^TW^{-1}MA\one},\nonumber\\
C_2(D^*)&\overset{\Delta}{=}: \frac{D^* + \one^TW^{-1}V}{\one^TW^{-1}MA\one}.\nonumber
\end{align}
\begin{proof}
After multiplying \eqref{cstr:steadState} by $\one^TW^{-1}M$, combining with \eqref{cstr:totWork} and some basic matrix manipulations the result is obtained. 
\end{proof}
\end{lem}

\begin{lem}
\label{lemma:D}
If \eqref{cstr:totWork} and \eqref{cstr:steadState} are satisfied, then there is a unique workload distribution which follows from the desired, chosen temperature distribution, i.e.
\begin{align}
D&=C_3\Tout + C_4(D^*)\label{eq:optimalWork}, \\[0.5ex]
C_3 &\overset{\Delta}{=}: -W^{-1}MA(I_n-\one C_1^T),\nonumber\\
C_4(D^*)&\overset{\Delta}{=}: W^{-1}MA\one C_2(D^*)-W^{-1}V.\nonumber
\end{align}
\begin{proof}
The proof of this lemma is along the same lines as the proof of \autoref{lemma:Tsup} and is therefore omitted.
\end{proof}
\end{lem}

\begin{rem}
The dimensions of the constants from above lemma's are $C_1~\in~\R^n$, $C_2 \in \R$, $C_3 \in \R^{n\times n}$ and $C_4 \in \R^n$. The following identities for the constants $C_1$, $C_3$ and $C_4$ are observed
\begin{equation}
C_1^T\one = 1, \quad \one^TC_3 = \zero, \quad C_3\one = \zero,\quad \one^TC_4 = D^*.
\label{eq:propConst}
\end{equation}
\hfill\QED
\end{rem}

\autoref{lemma:Tsup} and \autoref{lemma:D} show that at the steady state the supply temperature and workload distribution are uniquely defined by the total workload, $D^*$, and the temperature distribution, $\Tout$. With these properties in mind we are ready to prove \autoref{thm:reducedOpt}.
\begin{proof}[Proof of \autoref{thm:reducedOpt}]
As shown in \autoref{remark:totPower} for a data center with homogeneous racks the total computational power consumption depends only on $D^*$ and is independent of distribution $D$. Then from \eqref{eq:costFuncReduced} it is seen that the power consumption is only dependent on the supply temperature of the CRAC. Combining this with \autoref{assum:COP} and \autoref{lemma:Tsup} trivially concludes the proof. With \autoref{lemma:D}, \eqref{cstr:compCap} is written as \eqref{cstr:compCapReduced}.
\end{proof}

%%%%%%%%%%%%%%%%%%%%%%%%%%%%%%%%%%%%%%%%%%%%%%%%%%%%%%%%%%%%%%%%%%%%%%
% Section Characterization of optimal solution
%%%%%%%%%%%%%%%%%%%%%%%%%%%%%%%%%%%%%%%%%%%%%%%%%%%%%%%%%%%%%%%%%%%%%%

\section{Characterization of the optimal solution}
\label{sec:CharacOptSol}
In the previous section we have showed the possibility to reduce the optimization problem to a simpler form. In this section we show that using KKT optimality conditions it is possible to further characterize the optimal point. 
\subsection{KKT optimality conditions}
Because the optimization problem \eqref{eq:optimizationReduced} is convex and all inequality constraints are linear functions we have that Slater's condition holds. Therefore it follows that $\Tout[bar]$ is an optimal solution to \eqref{eq:optimizationReduced} if and only if there exists $\bar{\mu},\bar{\mu}_+,\bar{\mu}_- \in \R_{\geq0}^n$ such that the following set of relations is satisfied \cite{Boyd04Convex}:
\begin{subequations}
\label{eq:LagrangeConditions}
\begin{align}
-C_1 + \bar{\mu} + C_3^T(\bar{\mu}_+-\bar{\mu}_-) &= \zero,\label{eq:Lagrangian}\\
\zero\preccurlyeq C_3\Tout[bar] + C_4(D^*)&\preccurlyeq \Dmax,\\
\Tout[bar] &\preccurlyeq \Tsafe,\\
\bar{\mu}_+^T(C_3\Tout[bar] + C_4(D^*)-\Dmax) &= 0,\label{eq:compSlackMaxWork}\\
\bar{\mu}^T_-(-C_3\Tout[bar] - C_4(D^*)) &= 0,\label{eq:compSlackMinWork}\\
\bar{\mu}^T (\Tout[bar]-\Tsafe) &= 0,\label{eq:compSlackMaxT}\\
\bar{\mu},\bar{\mu}_+,\bar{\mu}_- &\succcurlyeq \zero.\label{eq:LagrangeMultipliers}
\end{align}
\end{subequations}
 
The Lagrangian corresponding to the optimal problem is given by:
\begin{align}
\mathcal{L}(\mu,\mu_+,\mu_-,\Tout) = &- C_1^T\Tout + \mu^T (\Tout-\Tsafe) \nonumber\\
&+ \mu^T_-(-C_3\Tout - C_4(D^*)) \label{eq:lagrangian}\\
&+ \mu_+^T(C_3\Tout + C_4(D^*)-\Dmax).\nonumber
\end{align}
\subsection{Lagrange multipliers}
By studying the KKT optimality conditions we can characterize the optimal solution in different cases. 
\begin{itemize}
	\item \emph{Inactive workload constraints:} Every rack is processing some work but not all the processors of each rack are utilized:
	\[0< (C_3\Tout[bar] + C_4(D^*))_i< \Dmax[i] \quad \forall i \in \{1,\cdots,n\}.\]
	\item \emph{Partially active workload constraints:} In $k$ racks all processors are processing jobs. The other $n-k$ racks are processing some work but still have processors available:
	\begin{align*} 
	(C_3\Tout[bar] + C_4(D^*))_i&= \Dmax[i] \quad \forall i \in \{1,\cdots,k\},\\ 
	0< (C_3\Tout[bar]+ C_4(D^*))_i&< \Dmax[i] \quad \forall i \in \{k+1,\cdots,n\}.
	\end{align*}
\end{itemize}
The characterization of these two cases is summarized in the following two theorems.
\begin{thm}
\label{thm:InactiveConstraints}
Assume the case that none of the workload constraints are active, i.e. 
\[0< (C_3\Tout[bar] + C_4(D^*))_i< \Dmax[i] \quad \forall i \in \{1,\cdots,n\}.\]
The solution to \eqref{eq:LagrangeConditions} and the optimal solution for the optimization problem \eqref{eq:optimizationReduced} is then given by
\be
\bar{\mu}_+ = \bar{\mu}_-=\zero, \quad\bar{\mu} = C_1 \succ\zero, \quad\Tout[bar] = \Tsafe.
\ee
\end{thm}
\begin{proof}
Because all the inequality constraints regarding the workload are inactive we have that both $C_3\Tout[bar] + C_4(D^*)-\Dmax~\prec~\zero$, and $-C_3\Tout[bar]~-~C_4(D^*)~\prec~\zero$. Then from \eqref{eq:compSlackMaxWork} and \eqref{eq:compSlackMinWork} we have that $\bar{\mu}_+ = \bar{\mu}_-=\zero$. From \eqref{eq:Lagrangian} it follows that $\bar{\mu} = C_1 \succ\zero$ such that from \eqref{eq:compSlackMaxT} we conclude that $\Tout[bar] = \Tsafe$. 
\end{proof}

\begin{thm}
\label{thm:partiallyActive}
In the case that a part of the workload constraints are active, i.e. 
\begin{align*} 
(C_3\Tout[bar] + C_4(D^*))_i&= \Dmax[i] \quad \forall i \in \{1,\cdots,k\},\\ 
0< (C_3\Tout[bar]+ C_4(D^*))_i&< \Dmax[i] \quad \forall i \in \{k+1,\cdots,n\},
\end{align*}
the solution of \eqref{eq:LagrangeConditions} is as follows:
\begin{enumerate}[(i)]
\item
For the racks at the constraint boundary, $i \in \{1,\cdots,k\}$:
\begin{align}
\bar{\mu}_-^i &= 0, \quad \frac{C^i_1 + \sum_{j=1,j\neq i}^k\bar{\mu}_+^j\abs{C_3^{ji}}}{C_3^{ii}} \geq \bar{\mu}_+^i \geq 0,\\ 
\bar{\mu}^i&=C^i_1 + \sum_{j=1,j\neq i}^k\bar{\mu}_+^j\abs{C_3^{ji}}-\bar{\mu}_+^iC_3^{ii}\geq 0, \\
\Tout[bar][i] &=\frac{\Dmax[i]- C^i_4(D^*)}{C_3^{ii}} + \sum_{j=k+1}^n{\frac{\abs{C_3^{ij}}}{C^{ii}_3}\Tsafe^j}\nonumber\\
&\quad+ \sum_{j=1,j\neq i}^k{\frac{\abs{C_3^{ij}}}{C^{ii}_3}\Tout[bar][j]}\nonumber\\
&\leq \Tsafe^i.
\end{align}
\item
For the racks that are not at the constraint boundary, $i\in\{k+1,\cdots,n\}$:
\begin{align}
\bar{\mu}_-^i &= \bar{\mu}_+^i = 0,\\ 
\bar{\mu}^i&=C^i_1 + \sum_{j=1}^k\bar{\mu}_+^j\abs{C_3^{ji}}>0, \\
\Tout[bar][i] &= \Tsafe^i.
\end{align}
\end{enumerate}
\hfill\QED
\end{thm}
Before we can prove \autoref{thm:partiallyActive} we need to know more about the structure of $C_3$. 
\begin{property}
\label{lemma:proofC3}
Consider $C_3$ as defined in \autoref{lemma:D}. The off-diagonal terms of this matrix are strictly negative and the diagonal terms are strictly positive.
\end{property}
\begin{proof}
The proof can be found in Appendix \ref{append:proofC3}.
\end{proof}
\begin{proof}[Proof of \autoref{thm:partiallyActive}]
Because part of the workload constraints are at the constraint boundary, the analysis following from the Lagrange multipliers is more involved. First we can say that 
\begin{align*}
\bar{\mu}_-^i&=0\quad\forall i,\\
\bar{\mu}_+^i&=0\quad\forall i\in\{k+1,\cdots,n\},\\
\bar{\mu}_+^i&\geq0\quad\forall i\in\{1,\cdots,k\}.
\end{align*}
Then from \eqref{eq:Lagrangian}
\begin{align}
\bar{\mu}^i&=C^i_1 - \sum_{j=1}^k\bar{\mu}_+^jC_3^{ji}.\label{eq:LagrangeMu}
\end{align}
From \autoref{lemma:proofC3} we have that the off-diagonal elements of $C_3$ are strictly negative. For racks $i \in \{k+1,\cdots,n\}$ we have that the $C_3^{ji}$ elements in \eqref{eq:LagrangeMu} will always be off-diagonal elements. Therefore rewriting \eqref{eq:LagrangeMu} gives
\begin{align}
\bar{\mu}^i=C^i_1 + \sum_{j=1}^k\bar{\mu}_+^j\abs{C_3^{ji}}>0\quad&\forall i \in \{k+1,\cdots,n\},\\
\Rightarrow \Tout[bar][i] = \Tsafe^i\quad&\forall i \in \{k+1,\cdots,n\}.
\end{align}
For racks $i \in \{1,\cdots,k\}$ \eqref{eq:LagrangeMu} is given by
\begin{align}
\bar{\mu}^i&=C^i_1 + \sum_{j=1,j\neq i}^k\bar{\mu}_+^j\abs{C_3^{ji}}-\bar{\mu}_+^iC_3^{ii}\geq 0, \\
\Rightarrow& \frac{C^i_1 + \sum_{j=1,j\neq i}^k\bar{\mu}_+^j\abs{C_3^{ji}}}{C_3^{ii}}\geq\bar{\mu}_+^i\quad\forall i \in \{1,\cdots,k\}.\label{eq:LagrangeMuPlus}
\end{align}
As the left hand side of \eqref{eq:LagrangeMuPlus} is strictly positive for all $i~\in~\{1,\cdots,k\}$, it is possible to find feasible $\bar{\mu}_+^i\geq0$ such that $\bar{\mu}^i \geq 0$ for all $i$. It can be shown that $\Tout[bar][i]$ for all $i~\in~\{1,\cdots k\}$ is given as
\begin{align}
\Tout[bar][i] = &\frac{\Dmax[i]- C^i_4(D^*)}{C_3^{ii}} + \sum_{j=k+1}^n{\frac{\abs{C_3^{ij}}}{C^{ii}_3}\Tsafe^j}\nonumber\\
&+ \sum_{j=1,j\neq i}^k{\frac{\abs{C_3^{ij}}}{C^{ii}_3}\Tout[bar][j]}\nonumber\\
&\leq \Tsafe^i.
\label{eq:ToutbarActiveConstraints}
\end{align}
\end{proof}
\begin{rem}
One cannot freely choose the $k$ racks for which $D_i = \Dmax[i]$. Whether or not a rack is processing its maximum capacity depends on the data center parameters, i.e. small amount of recirculated air at the input of the rack and low power consumption of the computational equipment. For these racks it holds that
\[
\Tout[bar][i] \leq \Tsafe^i \quad\forall i \in \{1,\cdots,k\}.
\]
\end{rem}

%%%%%%%%%%%%%%%%%%%%%%%%%%%%%%%%%%%%%%%%%%%%%%%%%%%%%%%%%%%%%%%%%%%%%%
% Section Temperature based job scheduling control
%%%%%%%%%%%%%%%%%%%%%%%%%%%%%%%%%%%%%%%%%%%%%%%%%%%%%%%%%%%%%%%%%%%%%%

\section{Temperature based job scheduling control}
\label{sec:Controller}
As established in \autoref{sec:CharacOptSol} it is possible to calculate the optimal solution under the assumption that the total workload at time $t$, $D^*$ is known. However it might not always be possible to obtain this quantity. For example when jobs arrive in the data center in some cases it might be hard to assess how much resources these jobs need. Consider the case where a virtual machine is requested by a user. Usually a certain amount of resources are allocated to this virtual machine, however the user need not use all the available resources all the time. In those situation it is hard to obtain the real workload. In this section we design a controller that is still able to achieve the control goals defined in \eqref{goal:Tout}-\eqref{goal:Work} under the assumption that $\zero\prec D\prec\Dmax$. From \autoref{thm:InactiveConstraints} we see that in this case the optimal solution is always $\Tout[bar] = \Tsafe$, independent of the way the jobs are distributed. Since most data centers are designed to have overcapacity usually the computational bounds of the racks will not be reached and this assumption is valid in those setups.

\subsection{Controller design} 
We will now design the control inputs for the workload distribution, $D$, and the supply temperature of the CRAC unit, $\Tsupscal$ while the total workload $D^*$ is unknown. Furthermore the controllers only have access to the measurement of the output temperature of the air at the outlet of each rack, $\Tout$. In other words we design temperature feedback algorithms to dynamically adjust $D$ and $\Tsupscal$ such that control objectives \eqref{goal:Tout}-\eqref{goal:Work} are achieved. The proposed controllers for the supply temperature and the workload distribution are given by 
\begin{align} 
\Tsup[dot][scal] &= \one^TA^TZ(\Tout-\Tsafe),\label{eq:controllerTsup}\\ 
\D[dot] &= (\frac{\one\one^T}{n}-I_n)B^TZ(\Tout-\Tsafe)\label{eq:controllerWork}, 
\end{align} 
where $A$ is Hurwitz, see Appendix \ref{append:Ahurwitz} for the proof of this property, $Z$ is the symmetric positive definite matrix such that 
\begin{equation} 
A^TZ+ZA = -2I_n, \label{eq:lyapEquation} 
\end{equation} 
and $B$ is
\[
B = M^{-1}W,
\]
where $W$ is defined \autoref{subsec:PowConsRacks}, and $A$ and $M$ are defined in \autoref{subsec:ThermModel}. The controllers \eqref{eq:controllerTsup} and \eqref{eq:controllerWork} depend only on the output temperature and the system parameters and will continue to vary until the output temperature reaches the safe temperature, which is in line with the control objectives. Intuitively the workload controller will shift jobs between racks based on the temperature deviation until the data center has reached the optimal state. In the results below we discuss the convergence behavior of the controllers in a time frame where the total workload, $D^*$, is assumed to be constant.
\begin{thm}
\label{thm:controllerConvergence} 
Assume $D^*$ is constant and $\one^T D(0) = D^*$. Then the solution of system \eqref{eq:ThermodynamicalModel} with controllers \eqref{eq:controllerTsup} and \eqref{eq:controllerWork} is bounded and converges to the optimal solution of the optimal problem defined in \eqref{eq:optimizationDef} and therefore satisfies control objectives \eqref{goal:Tout}-\eqref{goal:Work}.
\end{thm}
\begin{proof}
For ease of notation we introduce incremental variables to denote deviations from optimal values
\begin{align*}
\Tout[tilde] &= \Tout-\Tout[bar], \\
\Tsup[tilde][scal] &= \Tsupscal - \Tsup[bar][scal], \\
\tilde{D} &= D - \D[bar]
\end{align*}
With these definitions system \eqref{eq:ThermodynamicalModel} can be rewritten as
\begin{equation}
\dot{\tilde{T}}_\ttout = A\Tout[tilde] - A\Tsup[tilde] + B\tilde{D},\label{eq:controlModel} 
\end{equation}
where $A$ and $B$ are as before
\begin{align*}
A &= \rho c_pM^{-1}(\Gamma^T - I_n)F,\\
B &= M^{-1}W.
\end{align*}
Define the incremental storage functions as
\begin{align}
\Xi_1(\Tsup[tilde][scal]) &= \frac{1}{2}\norm{\Tsup[tilde][scal]}^2,\\
\Xi_2(\D[tilde]) &= \frac{1}{2}\norm{\D[tilde]}^2.
\end{align}
The storage functions satisfy
\begin{align}
\dot{\Xi}_1(\Tsup[tilde][scal],\Tout[tilde]) &= \Tsup[tilde][scal]^T\Tsup[dot][scal]\nonumber\\
&= \Tsup[tilde][scal]^T\one^TA^TZ\Tout[tilde], \label{eq:storFuncXi1dot}
\end{align}
and
\begin{align}
\dot{\Xi}_2(\D[tilde],\Tout[tilde]) &= \tilde{D}^T\D[dot]\nonumber\\
&= \tilde{D}^T(\frac{\one\one^T}{n}-I_n)B^TZ\Tout[tilde]\\
&= \tilde{D}^T\frac{\one\one^T}{n}B^TZ\Tout[tilde]-\tilde{D}^TB^TZ\Tout[tilde].\label{eq:storFuncXi2dot}
\end{align}
From constraint \eqref{cstr:totWork} we have that $\one^T D(t) = D^*$ should be satisfied for all $t\geq0$. First we notice that $\one^T \D[dot] = 0$ at all times $t\geq 0$. Hence if $\one^TD(0) = D^*$ then $\one^TD(t) = D^*$ for all $t \geq 0$. With this we see that $\tilde{D}^T\one=(D-\D[bar])^T\one=D^*-D^*=0$ such that \eqref{eq:storFuncXi2dot} is reduced to
\begin{equation}
\dot{\Xi}_2(\D[tilde],\Tout[tilde]) = -\tilde{D}^TB^TZ\Tout[tilde].
\label{eq:storFuncXi2dotReduced}
\end{equation}
Now consider the following Lyapunov function $V(\Tout[tilde]) = \frac{1}{2}\Tout[tilde]^TZ\Tout[tilde]$, where $Z$ is defined in \eqref{eq:lyapEquation}. Then $V(\Tout[tilde])$ satisfies
\begin{equation}
\dot{V}(\Tout[tilde]) = -\norm{\Tout[tilde]}^2 - \Tsup[tilde][scal]^T\one^TA^TZ\Tout[tilde] + \tilde{D}^TB^TZ\Tout[tilde].
\label{eq:lyapFunctionDot}
\end{equation}
Hence,
\begin{equation}
\dot{V}(\Tout[tilde]) + \dot{\Xi}_1(\Tsup[tilde][scal],\Tout[tilde]) + \dot{\Xi}_2(\D[tilde],\Tout[tilde]) = -\norm{\Tout[tilde]}^2 \leq 0.
\label{eq:totLyapFunction}
\end{equation}
Using LaSalle's invariance principle this result implies that every solution to the closed loop system initialized as $\one^TD(0) = D^*$ is bounded and shows convergence to the largest invariant set where $\Tout[tilde] = 0$. Because the temperature distribution converges to the optimal distribution and the total workload is given by $D^*$ we see from \autoref{lemma:Tsup} and \autoref{lemma:D} that the supply temperature and workload distribution will converge to a unique value given by \eqref{eq:optimalTsup} and \eqref{eq:optimalWork}. Because these values are unique these values are also optimal.   
Therefore we conclude that system \eqref{eq:controlModel} with controllers \eqref{eq:controllerTsup} and \eqref{eq:controllerWork} satisfies control objectives \eqref{goal:Tout}-\eqref{goal:Work}, and the state and the inputs of the system converge to the optimal solution.
\end{proof}

The proposed controller for the workload rebalances the workload currently present in the data center. The initial scheduling is assumed to be taken care of by an external entity over which we have no control. This approach is most applicable in cases where the scheduling is done in a non-controllable way, e.g. when the scheduling is hard-coded and incoming jobs are scheduled by means of chassis numbers. In these situations the only option available is to move jobs around to drive the data center to the the optimal state.  

%%%%%%%%%%%%%%%%%%%%%%%%%%%%%%%%%%%%%%%%%%%%%%%%%%%%%%%%%%%%%%%%%%%%%%
% Section Case study
%%%%%%%%%%%%%%%%%%%%%%%%%%%%%%%%%%%%%%%%%%%%%%%%%%%%%%%%%%%%%%%%%%%%%%
\section{Case study}
\label{sec:CaseStudy}
To evaluate the performance of the proposed controller, we use Matlab to simulate the closed loop system with a synthetic workload trace. For both the data center parameters and the workload trace we use the data presented in \cite{Vasic10Thermal}. The data center parameters were obtained from measurements by Vasic et al. at the IBM Zurich Research Laboratory. This data is to our best knowledge the most extensive characterization of the heat recirculation parameters of a data center. 

\subsection{Data center parameters}
The simulated data center consists of 30 homogeneous server racks, i.e. the power consumption characteristics, the safe temperature threshold and physical parameters are identical for all 30 racks. The rack model is a Dell PowerEdge 1855, with 10 dual-processor blade servers, i.e. a total of 20 CPU units. The power consumption of the racks is modeled by $P_i(t)~=~1728~+~145.5D_i(t)$ \cite{Tang06Thermal}. The safe threshold temperature is set at 30\degr{}. 
\begin{figure}
	\centering
		\includegraphics[width=1.00\linewidth]{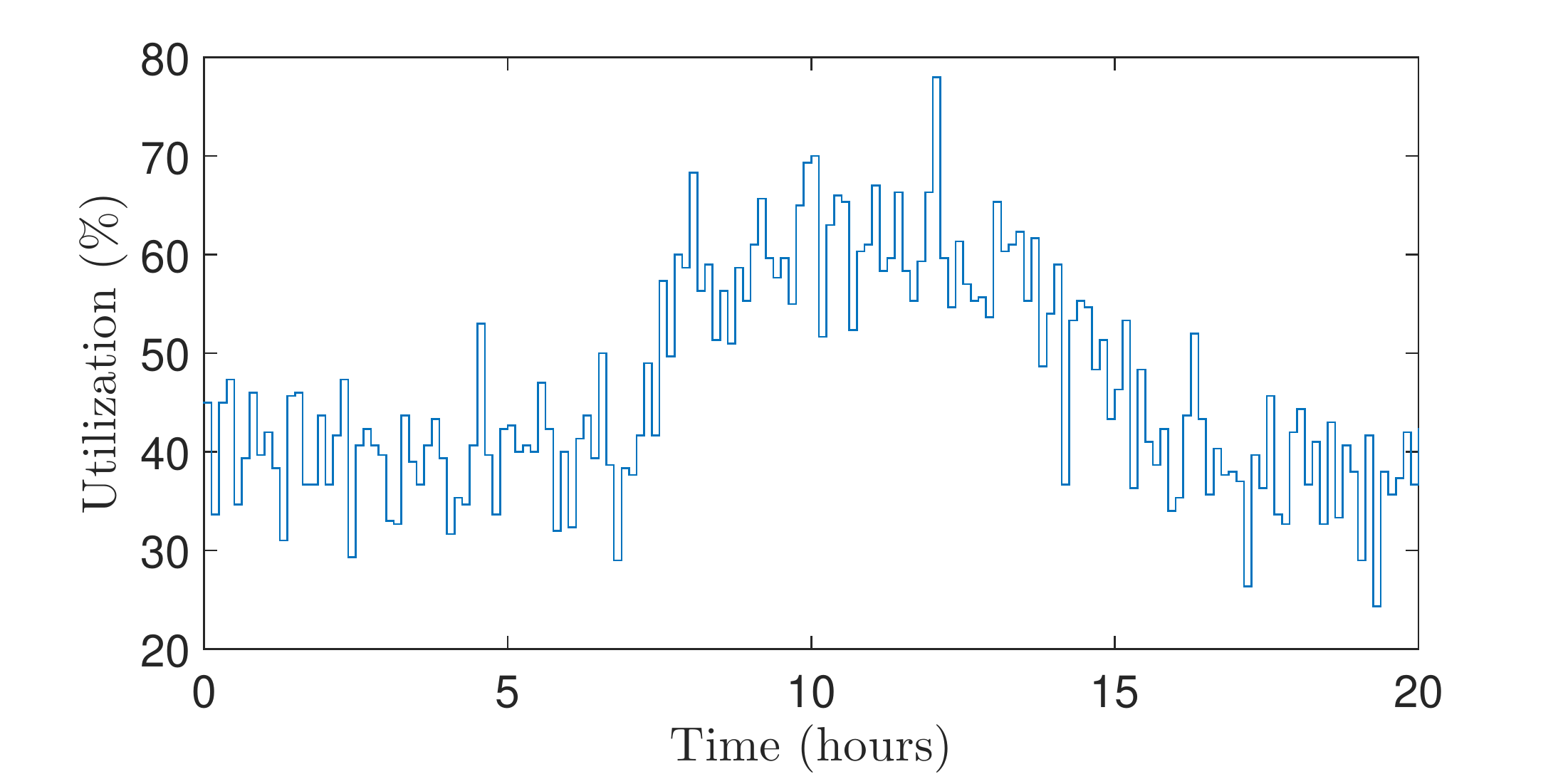}
	\caption{Synthetic workload trace supplied to data center. The workload varies $\pm 10\%$ about two nominal values, representing nighttime and daytime operation levels. The total workload changes every 7.5 minutes during which the workload is assumed to be constant.}
	\label{fig:totWorkloadLargeFont}
\end{figure}
\begin{figure}
	\centering
		\includegraphics[width=1.00\linewidth]{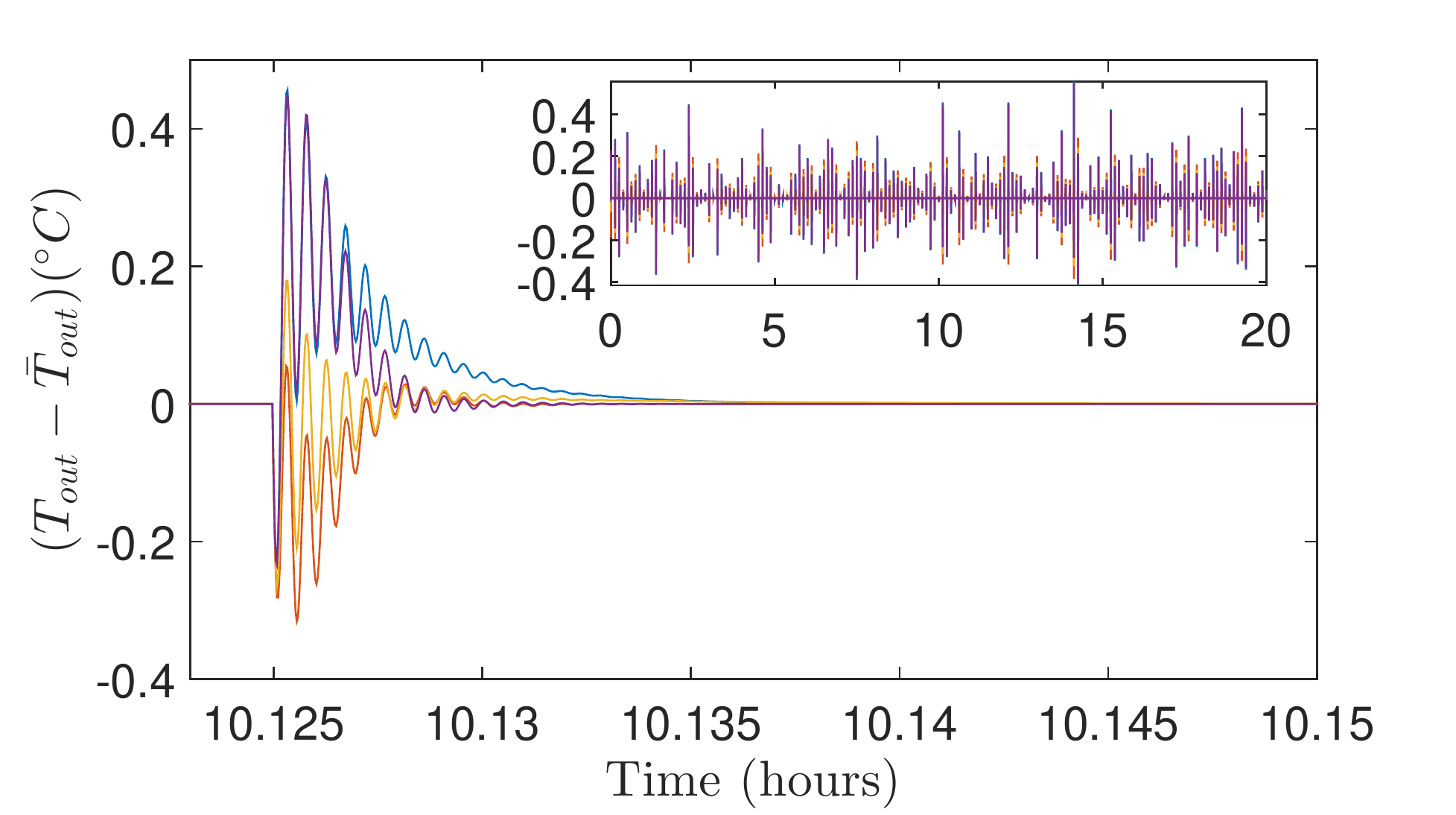}
	\caption{Plot of the response of $(\Tout - \Tout[bar])$ during the simulation for 4 selected racks. The full simulation is shown in the inset and the main plot is a magnification of the response after a change in total workload around $t = 10$ hours. Each time the total workload changes the temperature of the racks start to deviate from the optimal value, the controllers drive the data center to the new optimal solution, $(\Tout - \Tout[bar]) = \zero$ again. The oscillatory response of the output temperature coincides with the response of the supply temperature controller. Over the whole simulation the temperature is kept in a bandwidth of $\pm 0.5$~\degr{} around the target temperature distribution.}
	\label{fig:20160613-deltaTwithInset}
\end{figure}
\begin{figure}
	\centering
		\includegraphics[width=1.00\linewidth]{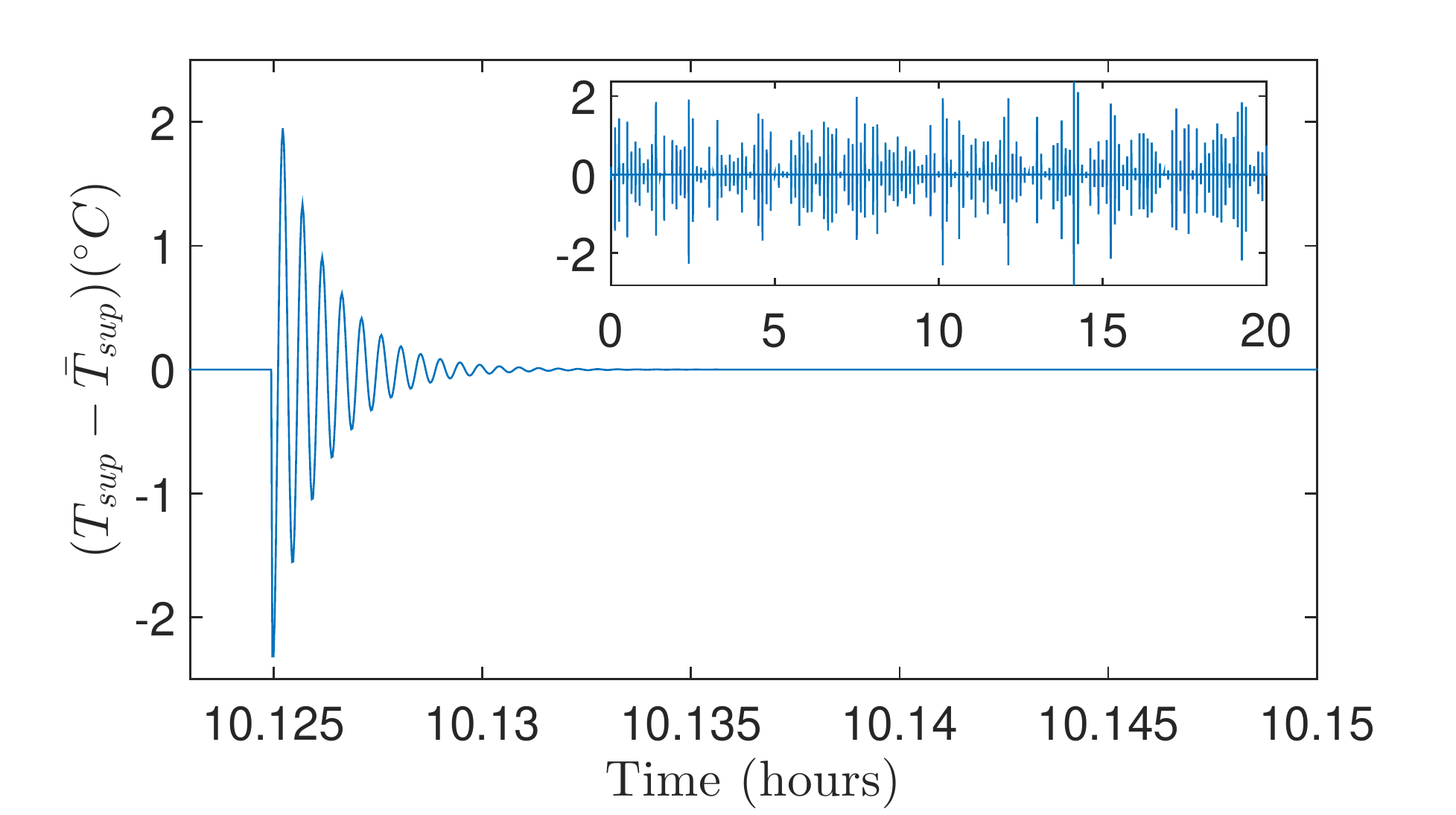}
	\caption{Plot of the response of $(\Tsupscal-\Tsup[bar][scal])$ during the simulation for 4 selected racks. The full simulation is shown in the inset and the main plot is a magnification of the response after a change in total workload around $t=10$ hours. The controller successfully drives the system to the new optimal value under varying total workload. The initial overshoot depends on the change of the total workload, i.e. the difference between the optimal supply temperatures in the two intervals. The oscillatory response results in an oscillatory fluctuation at the output temperature profile.}
	\label{fig:20160613-deltaTsupwithInset}
\end{figure}
\begin{figure}
	\centering
		\includegraphics[width=1.00\linewidth]{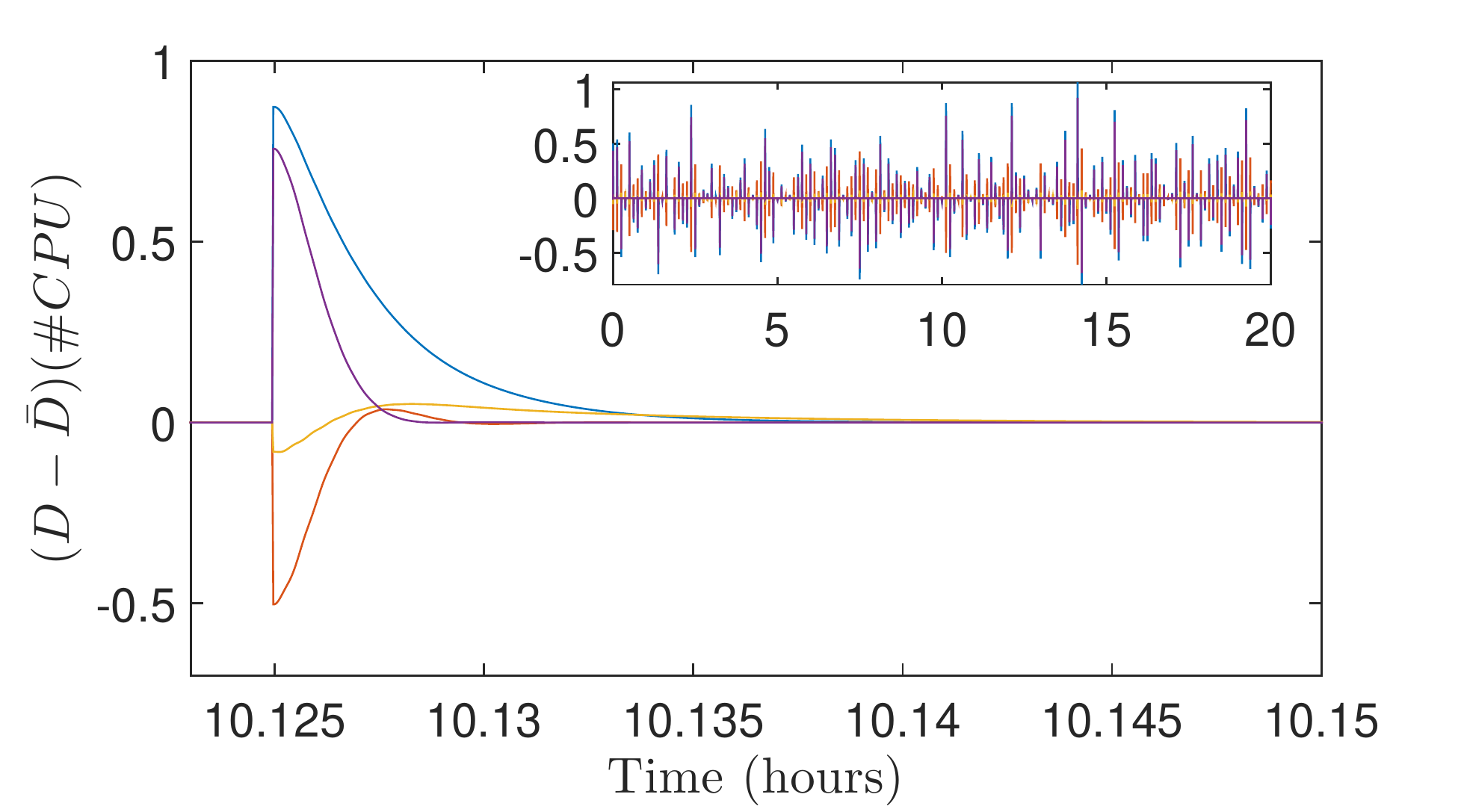}
	\caption{Plot of the response of $(D-\bar{D})$ during the simulation for 4 selected racks. The full simulation is shown in the inset and the main plot is a magnification of the response after a change in total workload around $t=10$ hours. As above the controller drives the system to the optimal value each time the total workload changes. When the total workload changes an external entity adds or subtracts work from the racks, this causes the initial overshoot. The controller redistributes the work again to the new optimal distribution.}
	\label{fig:20160613-deltaPwithInset}
\end{figure}
We supply a synthetic workload trace to the data center, see \autoref{fig:totWorkloadLargeFont}. The workload trace is constructed by varying the total workload by $\pm 10\%$ about two nominal values, $40\%$ and $60\%$ of the total data center capacity, representing nighttime and daytime operation levels respectively. The total workload is a piecewise constant function which changes value every 7.5 minutes.
Each time the total workload changes new work is added by or released to an external entity over which we assume to have no control. After this update has taken place we observe the change in temperature from the desired temperature profile. When $(\Tout-\Tout[bar])$ starts deviating from $0$ the controllers will act to respond to the changing conditions. 

In \autoref{fig:20160613-deltaTwithInset}, \autoref{fig:20160613-deltaTsupwithInset} and \autoref{fig:20160613-deltaPwithInset} the responses of $(\Tout-\Tout[bar])$, $(\Tsupscal-\Tsup[bar][scal])$ and $(D-\bar{D})$ respectively for 4 selected racks are shown. To investigate the performance of the controllers we calculated the optimal values for the variables offline and used those to plot the incremental variables. The initial overshoots in \autoref{fig:20160613-deltaTsupwithInset} and \autoref{fig:20160613-deltaPwithInset} are dependent on the change in total workload between intervals. The larger the change, the larger this initial overshoot will be. We observe different behavior for the two controllers. The controller for the supply temperature results in very oscillatory behavior for the supply temperature which in turn results in a fluctuating output temperature profile. The controller for the workload division however shows a much smoother response and more gradually steers the workload distribution to the optimal distribution. Every time the workload changes the controllers drive the system back to the optimal value in approximately $0.01$ hour = $36$ seconds. 

Although this is a very quick response it is not likely this convergence time will be attained in practice. In the simulation the cooled air of the CRAC instantly reaches the racks, whereas in a real data center it will take some time for the air to travel from the CRAC to the racks. On the contrary the workload division happens on a much shorter timescale, therefore we expect that in practice the output temperature will first increase, as new work is assigned to the rack, and after a certain delay the cooling will start to kick in to drive the temperature profile back to the setpoint. 

The supplied workload simulated a day and night cycle to study the response of the controller under large varying loads. From the results we see no difficulty for the controller to handle these different conditions. We conclude that the controller is able to keep the temperature of the racks around the target setpoint under all load conditions. 

%%%%%%%%%%%%%%%%%%%%%%%%%%%%%%%%%%%%%%%%%%%%%%%%%%%%%%%%%%%%%%%%%%%%%%
% Section Conclusions
%%%%%%%%%%%%%%%%%%%%%%%%%%%%%%%%%%%%%%%%%%%%%%%%%%%%%%%%%%%%%%%%%%%%%%

\section{Conclusions and future work}
\label{sec:Conclusions}
Many papers on thermal-aware job scheduling have studied the topic from a practical perspective however a theoretical analysis has less often been done. In this work we describe data centers and corresponding thermodynamics in a control theoretical fashion combining optimization theory with controller design. 

We have studied the minimization of energy consumption in a data center where recirculation of airflow is present, i.e. inefficiencies in cooling of the racks, through thermal-aware job scheduling and cooling control. We have set up an optimization problem and characterized the optimal workload distribution and cooling temperature to achieve minimum energy consumption while ensuring job processing and thermal threshold satisfaction. In addition we have presented a controller that works under varying workload conditions and is able to drive the control and state variables to the optimal values.

We have shown that it is possible to uniquely determine the optimal cooling supply temperature and workload distribution as a function of the total workload and desired temperature distribution of the racks in the data center. Furthermore we have shown that the optimal temperature distribution can be analytically calculated and that this distribution is independent of the workload distribution if none of the racks reaches its computational capacity. 

With the assumption that none of the racks is at its computational capacity we have designed controllers that control the supply temperature and workload distribution to drive the data center to the optimal state. 

There are several directions in which we want to extend our research. First we want to extend the framework to include situations where the optimal temperature distribution changes due to racks reaching their computational capacity. This will allow us to include server consolidation where the number of active racks is decreased to reduce energy consumption. In these situations it is inevitable that the computational capacity of the racks is reached and that varying optimal temperature distributions will have to be addressed.

Our control approach requires knowledge of the thermal characteristics of the data center. Studying the robustness and stability of our approach under small variations of the heat recirculation matrix is therefore of importance. Lastly it would be interesting to study the inclusion of other variables in the optimization problem, such as Service Level Agreements and response times of the jobs. 

%%%%%%%%%%%%%%%%%%%%%%%%%%%%%%%%%%%%%%%%%%%%%%%%%%%%%%%%%%%%%%%%%%%%%%
% Section Appendix
%%%%%%%%%%%%%%%%%%%%%%%%%%%%%%%%%%%%%%%%%%%%%%%%%%%%%%%%%%%%%%%%%%%%%%

\section{Appendix}

%%%%%%%%%%%%%%%%%%%%%%%%%%%%%%%%%%%%%%%%%%%%%%%%%%%%%%%%%%%%%%%%%%%%%%
% Section Appendix: Proof of C3
%%%%%%%%%%%%%%%%%%%%%%%%%%%%%%%%%%%%%%%%%%%%%%%%%%%%%%%%%%%%%%%%%%%%%%

\subsection{Proof of \autoref{lemma:proofC3}}
\label{append:proofC3}
From \autoref{lemma:D} we have that
\[
C_3 = -W^{-1}MA(I_n-\one C_1^T),
\]
where
\[
C_1^T = \frac{\one^TW^{-1}MA}{\one^TW^{-1}MA\one}.
\]
Defining a temporary variable $\alpha = W^{-1}MA$ we can write $C_3$ as
\[
C_3 = -\alpha + \frac{1}{\one^T\alpha\one}\alpha\one\one^T\alpha.
\]
The $ij$-th component of $C_3$ is then given by
\begin{equation}
C_3^{ij} = - \alpha_{ij} + \frac{\sum_{l=1}^n{\alpha_{il}}\sum_{k=1}^n{\alpha_{kj}}}{\sum_{l=1}^n{\sum_{k=1}^n{\alpha_{lk}}}}.
\label{eq:componentC3}
\end{equation}

From the definition of $\alpha$ we find that the $ij$-th component of $\alpha$ is given by
\begin{equation}
\alpha_{ij} = c_p\rho\frac{1}{w_i}(\gamma_{ji}-\delta_{ji})f_j,
\label{eq:componentAlpha}
\end{equation}
where $\delta_{ji}$ is the Kronecker delta, which is 1 if $i=j$ and 0 otherwise. To simplify the mathematics a little from now on, we assume that the data center consists of homogeneous racks, see \autoref{remark:totPower}. Combining \eqref{eq:componentAlpha} with \eqref{eq:componentC3} we have

\begin{align}
C_3^{ij} = &-c_p\rho\frac{1}{w}\bigg( (\gamma_{ji}-\delta_{ji})f_j \nonumber\\
&+ \frac{\left(f_i -\sum_{l=1}^n{\gamma_{li}f_l}\right)\left(f_j-\sum_{k=1}^n{\gamma_{jk}f_j}\right)}{\sum_{l=1}^n{(f_l-\sum_{k=1}^n{\gamma_{kl}f_k})}} \bigg).
\label{eq:componentC3withAlpha}
\end{align}
Although the big fraction in \eqref{eq:componentC3withAlpha} looks a bit daunting it is actually easy to conceptually understand it. The airflow at the inlet of the rack consists of two parts, air coming from the CRAC unit and air recirculating from other racks to the rack in question. At the outlet of the rack the airflow is composed of the air going back to the CRAC unit and the air recirculating from the rack in question to all the other racks. Looking closer at the nominator of \eqref{eq:componentC3withAlpha} we see that the first half is the air flowing from the CRAC unit to rack $i$, and the second half is the air flowing from rack $j$ to the CRAC unit. The denominator is the sum of the airflow each rack receives from the CRAC unit which is equal to the supplied airflow, $f_\ttsup$. In this way we can simplify \eqref{eq:componentC3withAlpha} to

\begin{equation}
C_3^{ij} = -c_p\rho\frac{1}{w}\bigg( (\gamma_{ji}-\delta_{ji})f_j + \frac{f_\text{(CRAC to i)}f_\text{(j to CRAC)}}{f_\ttsup} \bigg).
\label{eq:componentC3withAlphareduced}
\end{equation}

Now in the case that $i \neq j$ \eqref{eq:componentC3withAlphareduced} is reduced to 

\begin{equation}
C_3^{ij} = \underbrace{-c_p\rho\frac{1}{w}}_{<0}\bigg( \underbrace{\gamma_{ji}f_j}_{>0} + \underbrace{\frac{f_\text{(CRAC to i)}f_\text{(j to CRAC)}}{f_\ttsup}}_{>0} \bigg) < 0.
\label{eq:componentC3withAlphareducedInotJ}
\end{equation}
Here we see that the off-diagonal terms of $C_3$ are strictly negative.

As for the diagonal terms, $i = j$, we have
\begin{align}
C_3^{ii} &= c_p\rho\frac{1}{w}\bigg( (1-\gamma_{ii})f_i - \frac{f_\text{(CRAC to i)}f_\text{(i to CRAC)}}{f_\ttsup} \bigg),\nonumber\\
\Rightarrow&(1-\gamma_{ii})f_i = \underbrace{f_i - \sum_{l=1}^n{\gamma_{li}f_l}}_{f_\text{(CRAC to i)}} + \sum_{l=1,l\neq i}^n{\gamma_{li}f_l}, \nonumber\\
\Rightarrow C_3^{ii}&= \underbrace{c_p\rho\frac{1}{w}}_{>0}\bigg( \underbrace{\sum_{l=1,l\neq i}^n{\gamma_{li}f_l}}_{>0} \nonumber\\
&\quad+ \underbrace{f_\text{(CRAC to i)}}_{>0}\bigg(\underbrace{1-\frac{f_\text{(j to CRAC)}}{f_\ttsup}}_{>0} \bigg)\bigg) > 0.
\label{eq:componentC3withAlphareducedIisJ}
\end{align}
In \eqref{eq:componentC3withAlphareducedIisJ} we see that the diagonal terms of $C_3$ are strictly positive. This concludes the proof.
\hfill\QED
%%%%%%%%%%%%%%%%%%%%%%%%%%%%%%%%%%%%%%%%%%%%%%%%%%%%%%%%%%%%%%%%%%%%%%
% Section Appendix: A is Hurwitz
%%%%%%%%%%%%%%%%%%%%%%%%%%%%%%%%%%%%%%%%%%%%%%%%%%%%%%%%%%%%%%%%%%%%%%
\subsection{Proof of Hurwitz property of matrix A}
\label{append:Ahurwitz}

Matrix $A$ as defined in \autoref{subsec:ThermModel} is given by
\begin{equation}
A = \rho c_pM^{-1}(\Gamma^T - I_n)F.
\label{eq:A-appendix}
\end{equation}
Writing the matrix out in full gives
\begin{equation}
A = \rho \pmat{
\frac{\gamma_{11}-1}{m_1}f_1 & \frac{\gamma_{21}}{m_1}f_2 & \cdots & \frac{\gamma_{n1}}{m_1}f_n \\
\vdots & \ddots &  & \vdots \\
\vdots &  & \ddots & \vdots \\
\frac{\gamma_{1n}}{m_n}f_1 & \frac{\gamma_{2n}}{m_n}f_2 & \cdots & \frac{\gamma_{nn}-1}{m_n}f_n
}.
\label{eq:AfullAppendix}
\end{equation}

If we can show that matrix $A$ is strictly diagonal dominant and that the diagonal elements are negative then by the Gerschgorin circle theorem we have shown that matrix $A$ is Hurwitz. 

First we will prove strict diagonal dominance of matrix $A$. As stated in Appendix \ref{append:proofC3}, the airflow in a rack exists of two part, the recirculation of air from the other racks and supplied air by the CRAC
\begin{align}
f_i &= \gamma_{ii}f_i + \sum_{j=1,j\neq i}^n{\gamma_{ji}f_j} + f_\ttsup^i,\nonumber\\
\Rightarrow (\gamma_{ii}-1)f_i &= -\sum_{j=1,j\neq i}^n{\gamma_{ji}f_j} - f_\ttsup^i \nonumber\\
&< -\sum_{j=1,j\neq i}^n{\gamma_{ji}f_j},\nonumber\\
\Rightarrow \abs{(\gamma_{ii}-1)f_i} &> \abs{-\sum_{j=1,j\neq i}^n{\gamma_{ji}f_j}} = \sum_{j=1,j\neq i}^n{\gamma_{ji}f_j}. \label{eq:RowDominanceA}
\end{align}
Because all $\gamma_{ij}$ are strictly between 0 and 1 both sides of the inequality are negative, the inequality sign changes when taking the absolute value of both sides. Comparing \eqref{eq:RowDominanceA} with \eqref{eq:AfullAppendix} and ignoring the mass, as the same mass appears in every element of row $i$, we see that the left hand side of \eqref{eq:RowDominanceA} is the same as the diagonal entry of $A$ and the right hand side is the same as the sum of the off-diagonal entries. As this holds for every row, then by the definition of strict diagonal dominant matrices, this shows that matrix $A$ is strictly diagonal dominant.

Furthermore as $\gamma_{ii}$ is strictly between 0 and 1, we have that all the diagonal elements of $A$ are strictly negative. Combining the above results with Gerschgorin circle theorem we have shown that all eigenvalue of matrix $A$ are strictly negative and therefore the matrix is Hurwitz.\hfill\QED

%%%%%%%%%%%%%%%%%%%%%%%%%%%%%%%%%%%%%%%%%%%%%%%%%%%%%%%%%%%%%%%%%%%%%%
% Acknowledgment
%%%%%%%%%%%%%%%%%%%%%%%%%%%%%%%%%%%%%%%%%%%%%%%%%%%%%%%%%%%%%%%%%%%%%%

\section*{Acknowledgment}
This research was carried out as part of the perspective program Robust Design of Cyber-Physical Systems, Cooperative Networked Systems and is supported by Technology Foundation STW and industrial partners Target Holding and Better.be. The authors would also like to thank IBM Zurich Research Lab for supplying measurement data of a real-life data center.

%%%%%%%%%%%%%%%%%%%%%%%%%%%%%%%%%%%%%%%%%%%%%%%%%%%%%%%%%%%%%%%%%%%%%%%%%%%%%%%%

\bibliographystyle{myIEEEtranNoURL}

\bibliography{../../Literature/ReferenceLib}

\begin{thebibliography}{10}
\providecommand{\url}[1]{#1}
\csname url@rmstyle\endcsname
\providecommand{\newblock}{\relax}
\providecommand{\bibinfo}[2]{#2}
\providecommand\BIBentrySTDinterwordspacing{\spaceskip=0pt\relax}
\providecommand\BIBentryALTinterwordstretchfactor{4}
\providecommand\BIBentryALTinterwordspacing{\spaceskip=\fontdimen2\font plus
\BIBentryALTinterwordstretchfactor\fontdimen3\font minus
  \fontdimen4\font\relax}
\providecommand\BIBforeignlanguage[2]{{%
\expandafter\ifx\csname l@#1\endcsname\relax
\typeout{** WARNING: IEEEtran.bst: No hyphenation pattern has been}%
\typeout{** loaded for the language `#1'. Using the pattern for}%
\typeout{** the default language instead.}%
\else
\language=\csname l@#1\endcsname
\fi
#2}}

\bibitem{DCD14JournalJanuaryFebruaryPage1617}
\BIBentryALTinterwordspacing
D.~Blatch, ``Is the industry getting better at using power?'' \emph{Datacenter
  Dynamics Focus}, vol.~3, no.~33, pp. 16--17, Jan/Feb 2014.
\BIBentrySTDinterwordspacing

\bibitem{ElectricityConsumptionWorld}
\BIBentryALTinterwordspacing
Enerdata, ``Global domestic electricity consumption,''
  https://yearbook.enerdata.net/electricity-domestic-consumption-data-by-region.html,
  August 2016.
\BIBentrySTDinterwordspacing

\bibitem{Hameed14Survey}
A.~Hameed, A.~Khoshkbarforoushha, R.~Ranjan, P.~P. Jayaraman, J.~Kolodziej,
  P.~Balaji, S.~Zeadally, Q.~M. Malluhi, N.~Tziritas, A.~Vishnu, S.~U. Khan,
  and A.~Zomaya, ``A survey and taxonomy on energy efficient resource
  allocation techniques for cloud computing systems,'' \emph{Computing}, pp.
  1--24, jun 2014.

\bibitem{Moore05Making}
\BIBentryALTinterwordspacing
J.~Moore, J.~Chase, R.~Parthasarathy, and S.~Ratnesh, ``Making scheduling
  'cool' temperature-aware workload placement in data centers,'' in
  \emph{USENIX Annual Technical Conference}, 2005, pp. 61--74.
\BIBentrySTDinterwordspacing

\bibitem{Tang08Energy}
Q.~Tang, S.~K.~S. Gupta, and G.~Varsamopoulos, ``Energy-efficient thermal-aware
  task scheduling for homogeneous high-performance computing data centers: a
  cyber-physical approach,'' \emph{IEEE trans. on Parallel and distributed
  systems}, vol.~19, pp. 1458--1472, nov 2008.

\bibitem{Sun14Energy}
H.~Sun, P.~Stolf, J.-M. Pierson, and G.~Da~Costa, ``Energy-efficient and
  thermal-aware resource management for heterogeneous datacenters,''
  \emph{Sustainable Computing: Informatics and Systems}, vol.~4, no.~4, pp.
  292--306, 2014.

\bibitem{Jiang14Thermal}
X.~Jiang, M.~I. Alghamdi, M.~M. Al~Assaf, X.~Ruan, J.~Zhang, M.~Qiu, and
  X.~Qin, ``Thermal modeling and analysis of cloud data storage systems,''
  \emph{Journal of Communications}, vol.~9, pp. 299--311, Apr 2014.

\bibitem{Pahlavan14Power}
\BIBentryALTinterwordspacing
A.~Pahlavan, M.~Momtazpour, and M.~Goudarzi, ``Power reduction in hpc data
  centers: a joint server placement and chassis consolidation approach,''
  \emph{The Journal of Supercomputing}, vol.~70, no.~2, pp. 845--879, 2014.
\BIBentrySTDinterwordspacing

\bibitem{Pakbaznia10Temperature}
\BIBentryALTinterwordspacing
E.~Pakbaznia, M.~Ghasemazar, and M.~Pedram, ``Temperature-aware dynamic
  resource provisioning in a power-optimized datacenter,'' in \emph{Proceedings
  of the Conference on Design, Automation and Test in Europe}.\hskip 1em plus
  0.5em minus 0.4em\relax European Design and Automation Association, 2010, pp.
  124--129.
\BIBentrySTDinterwordspacing

\bibitem{Li12Joint}
S.~Li, H.~Le, N.~Pham, J.~Heo, and T.~Abdelzaher, ``Joint optimization of
  computing and cooling energy: Analytic model and a machine room case study,''
  in \emph{32nd Int. Conf. on Distributed Computing Systems}.\hskip 1em plus
  0.5em minus 0.4em\relax IEEE, Jun 2012, pp. 396--405.

\bibitem{Tang06Thermal}
Q.~Tang, S.~K.~S. Gupta, D.~Stanzione, and P.~Cayton, ``Thermal-aware task
  scheduling to minimize energy usage of blade server based datacenters,'' in
  \emph{2nd IEEE Int. Symp. on Dependable, Autonomic and Secure
  Computing}.\hskip 1em plus 0.5em minus 0.4em\relax IEEE, Sep 2006, pp.
  195--202.

\bibitem{Vasic10Thermal}
N.~Vasic, T.~Scherer, and W.~Schott, ``Thermal-aware workload scheduling for
  energy efficient data centers,'' in \emph{Proceedings of the 7th
  international conference on Autonomic computing}, 2010, pp. 169--174.

\bibitem{Yin14Adaptive}
X.~Yin and B.~Sinopoli, ``Adaptive robust optimization for coordinated capacity
  and load control in data centers,'' in \emph{Decision and Control (CDC), 2014
  IEEE 53rd Annual Conference on}.\hskip 1em plus 0.5em minus 0.4em\relax IEEE,
  2014, pp. 5674--5679.

\bibitem{Doyle13Distributed}
J.~Doyle, F.~Knorn, D.~O'Mahony, and R.~Shorten, ``Distributed thermal aware
  load balancing for cooling pf modular data centres,'' \emph{IET Control
  Theory and Applications}, vol.~7, pp. 612--622, Nov 2013.

\bibitem{Albea14Hybrid}
\BIBentryALTinterwordspacing
C.~Albea, A.~Seuret, and L.~Zaccarian, ``Hybrid control of a three-agent
  network cluster,'' in \emph{53th IEEE Conference on Decision and Control},
  dec 2014, pp. 5302--5307.
\BIBentrySTDinterwordspacing

\bibitem{Parolini12Cyber}
L.~Parolini, B.~Sinopoli, B.~H. Krogh, and Z.~Wang, ``A cyber-physical systems
  approach to data center modeling and control for energy efficiency,''
  \emph{Proceedings of the IEEE}, vol. 100, pp. 254--268, jan 2012.

\bibitem{Burger15Dynamic}
\BIBentryALTinterwordspacing
M.~B{\"u}rger and C.~De~Persis, ``Dynamic coupling design for nonlinear output
  agreement and time-varying flow control,'' \emph{Automatica}, vol.~51, pp.
  210--222, 2015.
\BIBentrySTDinterwordspacing

\bibitem{Tang06SensorBased}
\BIBentryALTinterwordspacing
Q.~Tang, T.~Mukherjee, S.~K.~S. Gupta, and P.~Cayton, ``Sensor-based fast
  thermal evaluation model for energy efficient high-performance datacenters,''
  in \emph{Fourth International Conference on Intelligent Sensing and
  Information Processing, 2006. ICISIP 2006}.\hskip 1em plus 0.5em minus
  0.4em\relax IEEE, December 2006, pp. 203--208.
\BIBentrySTDinterwordspacing

\bibitem{Heath06Mercury}
\BIBentryALTinterwordspacing
T.~Heath, A.~P. Centeno, P.~George, L.~Ramos, Y.~Jaluria, and R.~Bianchini,
  ``Mercury and freon: temperature emulation and management for server
  systems,'' in \emph{12th international conference on Architectural support
  for programming languages and operating systems}.\hskip 1em plus 0.5em minus
  0.4em\relax ASPLOS XII, 2006, pp. 106--116.
\BIBentrySTDinterwordspacing

\bibitem{Ranganathan06Ensemble}
\BIBentryALTinterwordspacing
P.~Ranganathan, P.~Leech, I.~David, and C.~J. S., ``Ensemble-level power
  management for dense blade servers,'' in \emph{33rd annual international
  symposium on Computer Architecture}.\hskip 1em plus 0.5em minus 0.4em\relax
  ISCA, 2006, pp. 66--77.
\BIBentrySTDinterwordspacing

\bibitem{Boyd04Convex}
S.~Boyd and L.~Vandenberghe, \emph{Convex optimization}.\hskip 1em plus 0.5em
  minus 0.4em\relax Cambridge university press, 2004.

\end{thebibliography}

\end{document}